\documentclass[11pt]{article}

\usepackage{amsmath,amssymb,amsfonts,amsthm,fullpage,graphicx}

\newtheorem{theorem}{Theorem}
\newtheorem{definition}{Definition}
\newtheorem{corollary}[theorem]{Corollary}    
\newtheorem{lemma}[theorem]{Lemma}    
\renewenvironment{proof}{\noindent{\bf Proof:}}{\qed\\}
\newenvironment{proofof}[1]{\noindent{\bf Proof of #1:}}{\qed\\}

\newcommand{\complex}{{\mathbb C}}
\newcommand{\reals}{{\mathbb R}}

\newcommand{\ket}[1]{| #1 \rangle}
\newcommand{\bra}[1]{ \langle #1 |}
\newcommand{\ketbra}[2]{| #1 \rangle\!\langle #2 |}
\newcommand{\braket}[2]{\langle #1 | #2 \rangle }
\newcommand{\norm}[1]{\left\| #1 \right\|}
\newcommand{\trnorm}[1]{\left\| #1 \right\|_{\mathrm{tr}}}
\newcommand{\size}[1]{\left| #1 \right|}
\newcommand{\set}[1]{\left\{ #1 \right\}}

\newcommand{\trace}{{\mathrm{Tr}}}
\newcommand{\support}{{\mathrm{supp}}}

\newcommand{\density}[1]{\ketbra{#1}{#1}}
\newcommand{\transpose}{{\mathrm T}}

\newcommand{\complexi}{{\mathrm{i}}}

\newcommand{\tr}{\mathsf{\trace}}
\newcommand{\defeq}{\stackrel{\scriptsize{\mathsf{def}}}{=}}
\newcommand{\eqdef}{\stackrel{\scriptsize{\mathsf{def}}}{=}}

\newcommand{\adjoint}{\dagger}

\newcommand{\suppress}[1]{}
\newcommand{\comment}[1]{}

\newcommand{\etal}{\emph{et al.\/}}

\newcommand{\cH}{{{\mathcal H}}}
\newcommand{\cK}{{{\mathcal K}}}
\newcommand{\sD}{{{\mathsf D}}}
\newcommand{\sS}{{{\mathsf S}}}
\newcommand{\sF}{{{\mathsf F}}}

\newcommand{\rI}{{{\mathrm I}}}
\newcommand{\rL}{{{\mathrm L}}}
\newcommand{\eps}{\varepsilon}

\newcommand{\rD}{{{\mathrm D}}}
\newcommand{\rP}{{{\mathrm P}}}

\title{ {\bf Short proofs of the Quantum Substate Theorem}
}

\author{
Rahul Jain\thanks{
Centre for Quantum Technologies and Department of Computer Science, 
National University of Singapore,  Block S15, 3 Science Drive~2,
Singapore 11754.
Email: \texttt{rahul@comp.nus.edu.sg}.
Work done in part while visiting the Institute for Quantum Computing,
University of Waterloo.
} \\
National U.\ Singapore
\and
Ashwin Nayak\thanks{
Department of Combinatorics and Optimization,
and Institute for Quantum Computing, University of Waterloo,
200 University Ave.\ W.,
Waterloo, ON, N2L 3G1, Canada.
Email: {\tt ashwin.nayak@uwaterloo.ca}.
Work done in part at Center for Quantum Technologies,
National University of Singapore, and at Perimeter Institute for
Theoretical Physics.
Research supported in part by NSERC Canada, CIFAR, an ERA (Ontario),
QuantumWorks, MITACS, and ARO (USA). Research at
Perimeter Institute is supported in part by the Government of Canada
through Industry Canada and by the Province of Ontario through MRI.
} \\
U.\ Waterloo
}

\date{December 30, 2011}

\begin{document}



\maketitle

\begin{abstract}
  The Quantum Substate Theorem due to Jain, Radhakrishnan, and
  Sen~(2002) gives us a powerful operational interpretation
  of relative entropy, in fact, of the observational divergence of two
  quantum states, a quantity that is related to their relative
  entropy. Informally, the theorem states that if the observational
  divergence between two quantum states~$\rho, \sigma$ is small, then
  there is a quantum state~$\rho'$ close to $\rho$ in trace distance,
  such that $\rho'$ when scaled down by a small factor becomes a
  substate of $\sigma$. We present new proofs of this theorem. The
  resulting statement is optimal up to a constant factor in its
  dependence on observational divergence.  In addition, the proofs are
  both conceptually simpler and significantly shorter than the earlier
  proof.
\end{abstract}

\comment{
IEEE IT keyword: quantum computing



MSC classification: 
Information and communication, circuits
94A17   	Measures of information, entropy
94A15   	Information theory, general
Quantum Theory
81P45   	Quantum information, communication, networks
81P94		Quantum cryptography
Computer Science
68P30   	Coding and information theory
68Q12   	Quantum algorithms and complexity
Operator theory
47B65   	Positive operators and order-bounded operators
}


\section{The Quantum Substate Theorem}
\label{sec-substate}

Consider quantum states~$\rho,\sigma \in \rD(\cH)$, where~$\cH$ is a
finite dimensional Hilbert space~$\cH$, and~$\rD(\cH)$ denotes the set
of all quantum states with support in~$\cH$, i.e., the set of unit
trace positive semi-definite operators on~$\cH$. We say that~$\rho$ is
a~$c$-substate of~$\sigma$ if~$\rho \preceq 2^c \sigma$,
where~$\preceq$ represents the L\"owner partial order on operators
on~$\cH$. We may equivalently express this condition in terms of
measurement outcomes (``POVM elements'') as follows. Let
\[
\rP(\cH) \quad \defeq \quad
    \set{ M \in \rL(\cH) \;:\; O \preceq M \preceq \rI } \enspace,
\]
denote the set of POVM elements on~$\cH$, where~$\rL(\cH)$ is the
space of linear operators and~$\rI$ is the identity operator on~$\cH$.
The state~$\rho$ is a~$c$-substate of~$\sigma$ iff for every
measurement outcome~$M \in \rP(\cH)$, the probability~$\tr(M\sigma)$
of observing~$M$ when~$\sigma$ is measured according to the
POVM~$\set{M, \rI- M}$ is at least~$\tr(M\rho)/2^c$, a~$1/2^{c}$
fraction of the probability of observing~$M$ when~$\rho$ is
measured. Morally, the state~$\sigma$ may be decomposed as~$\sigma =
\alpha \rho + (1-\alpha) \tilde{\sigma}$, for some~$\tilde{\sigma}
\in \rD(\cH)$,
with~$\alpha \geq 1/2^c$. This in turn may be used to construct the
state~$\rho$ from~$\sigma$ through quantum analogues of rejection
sampling. For example, we may apply the quantum measurement specified by
the Kraus operators~$\set{ \sqrt{\alpha}\, \rho^{1/2} \sigma^{-1/2},
\sqrt{1-\alpha} \, \tilde{\sigma}^{1/2} \sigma^{-1/2} }$, or go through a
purification of~$\sigma$~\cite{JainRS02,JainRS09}.

Given arbitrary quantum states~$\rho, \sigma \in \rD(\cH)$ we are
interested in how well~$\sigma$ masquerades as~$\rho$ in the above
sense. In other words, we are interested in the least~$c$ such
that~$\rho$ is a~$c$-substate of~$\sigma$. We call this quantity the
\emph{relative min-entropy\/}~$\sS_\infty(\rho \| \sigma)$ of the two
states. A generalization of this notion to bipartite states has been
studied by Renner~\cite[Chapter~3]{Renner05}, and the notion itself
has been studied by Datta~\cite{Datta09} as ``max-relative entropy''.
For typical applications, such as privacy trade-offs in communication
protocols~\cite{JainRS02,JainRS09}, it suffices to construct an
approximation~$\rho'$ to~$\rho$, with respect to a metric on quantum
states. This leads us to the notion of the \emph{smooth relative
  min-entropy\/}~$\sS^\eps_\infty(\rho \| \sigma)$ of the two states,
a quantity implicitly studied by Jain, Radhakrishnan, and
Sen~\cite{JainRS02,JainRS09} and later explicitly by
Renner~\cite[Chapter~3]{Renner05} and Datta~\cite{Datta09}.  The
metric initially used for the smoothness parameter~$\eps$ was the
trace distance. The fidelity of quantum states gives us a more natural
metric in typical applications, and we adopt this measure of closeness
in the article.

Let $\eps \in (0,1)$ and~$\rho, \sigma \in \rD(\cH)$ be such that
$\support(\rho) \subseteq \support(\sigma)$.  We may express
the~$\eps$-smooth relative min-entropy~$\sS^\eps_\infty(\rho \|
\sigma)$ as the base-$2$ logarithm of the value of the following
optimization problem with variables~$\rho' \in \rD(\cH)$ 
and~$\kappa \in \reals$:
\begin{align*}
\text{minimize:} & \quad \kappa \\
   \text{subject to:} & \\
\rho' \quad & \preceq \quad \kappa\, \sigma \\
\tr \, \rho' \quad & = \quad 1  & \textrm{(P1)} \\
\sF(\rho,\rho') \quad & \geq \quad 1 - \eps  \\
\rho' \in \rL(\cH), & \quad \rho' \succeq 0 \\
\kappa \in \reals, & \quad \kappa \geq 0 
\end{align*}
Here~$\sF(\rho', \rho) \eqdef \trnorm{ \sqrt{\rho'} \sqrt{\rho} }^2$,
denotes the fidelity between the two quantum states, and~$\trnorm{M}
\eqdef \tr \sqrt{M^\adjoint M}$ denotes the trace norm of the linear
operator~$M \in \rL(\cH)$.  The existence of a pair~$\rho',\kappa$
that are feasible for the problem (P1) means that there is a quantum
state~$\rho'$ with fidelity~$\sF(\rho', \rho) \geq 1- \eps$ that is
also a~$(\log_2 \kappa)$-substate of~$\sigma$. The substate constraint
implies that~$\kappa \geq 1$.

The program (P1) is feasible, as~$\rho' \eqdef \rho$ and~$\kappa \eqdef
1/\lambda$, where~$\lambda$ is the smallest non-zero eigenvalue
of~$\sigma$, satisfy all the constraints. Therefore we may restrict
the optimization to~$\kappa \in [0,1/\lambda]$ and the compact set of
quantum states with fidelity at least~$1-\eps$
with~$\rho$. The~$\eps$-smooth relative min-entropy between the two
states is thus always achieved.

If~$\rho$ is a~$c$-substate of~$\sigma$, i.e., their relative
min-entropy is at most~$c$, then their relative entropy $\sS(\rho \|
\sigma) \defeq \tr\, \rho (\log_2 \rho - \log_2 \sigma)$ is also at
most~$c$. Jain \etal~\cite{JainRS02,JainRS09} gave a weak converse to
this relation via the Quantum Substate Theorem, which gives a bound on
the~$\eps$-smooth relative min-entropy in terms of the more familiar
notion of relative entropy.  This theorem may also be viewed as a
handy operational interpretation of the rather abstract notion of
relative entropy.

The substate theorem (classical or quantum) lies at the heart of a
growing number of applications~\cite[Section~1]{JainRS09}. These
include privacy trade-offs in communication protocols for computing
relations~\cite{JainRS05}, message compression leading to direct sum
theorems in classical and quantum communication
complexity~\cite{JainRS05}, impossibility results for bit-string
commitment~\cite{Jain08}, the communication complexity of remote state
preparation~\cite{Jain06}, and direct product theorems for classical
communication complexity~\cite{JainKN08,Jain11}.  To highlight one of
these examples, the Quantum Substate Theorem enables (non-oblivious)
compression of an ensemble of mixed quantum states to within a
constant factor of the Holevo information of the ensemble, given
access to shared entanglement and classical communication, when we are
allowed a small loss of fidelity in the compression process. In
contrast, the compression of arbitrary ensembles of mixed quantum
states to the Holevo limit remains an open problem in quantum
information theory.

Jain \etal{} formulated their bound in terms of a new information
theoretic quantity, \emph{observational divergence\/}~$\sD(\rho \|
\sigma)$, rather than relative entropy.
\begin{definition}[Observational divergence]
\label{def:div}
Let $\rho, \sigma \in \rD(\cH)$. Their \emph{observational
  divergence\/} is defined as
\begin{displaymath}
\sD(\rho \| \sigma) \quad \defeq \quad 
    \sup \left\{ 
        (\tr \, M \rho) \log_2 \frac{\tr \, M \rho}{\tr \, M \sigma} 
        \;:\;  M \in \rP(\cH), \; \tr \, M \sigma  \neq 0
    \right\} \enspace.
\end{displaymath}
\end{definition}
The supremum in the definition above is achieved if and only
if~$\support(\rho) \subseteq \support(\sigma)$. As is evident, this
quantity is a scaled measure of the maximum factor by
which~$\tr(M\rho)$ may exceed~$\tr(M\sigma)$ for any measurement
outcome~$M$ of interest.  Observational divergence is related to
relative entropy. In particular, $\sD(\rho \| \sigma) \leq \sS(\rho \|
\sigma) + 1$. However, it could be smaller than relative entropy by a
factor proportional to the dimension~\cite[Proposition~4]{JainRS09}
(see also~\cite{JainNS10}).

We present alternative proofs of the Quantum Substate Theorem, also
strengthening it in the process.
\begin{theorem}
\label{thm-substate}
Let~$\cH$ be a Hilbert space, and let $\rho, \sigma \in \rD(\cH)$ be
quantum states such that $\support(\rho) \subseteq \support(\sigma)$.
For any~$\eps \in (0,1)$, there is a quantum state~$\rho'$ with
fidelity~$\sF(\rho', \rho) \geq 1 - \eps$ such that~$\rho' \preceq
\kappa \sigma$, where
\[
\kappa \quad = \quad \frac{1}{1-\eps}\; 2^{\sD(\rho\|\sigma)/\eps}
\enspace.
\]
Equivalently, for any~$\eps \in (0,1)$,
\[
\sS^\eps_\infty(\rho \| \sigma) \quad \leq \quad 
    \frac{\sD(\rho \| \sigma)}{\eps} + \log_2 \frac{1}{1-\eps} \enspace.
\]
\end{theorem}

The proofs that we present are both shorter and conceptually simpler
than the original proof. The proof due to Jain \etal{} consists of a
number of technical steps, several of which are bundled into a
``divergence lifting'' theorem that reduces the problem to one in
which~$\rho$ is a pure state (a rank one quantum state). Finally, the
pure state case is translated into a problem in two dimensions which
is solved by a direct calculation. Divergence lifting involves going
from a construction of a suitable state~$\rho'$ for a fixed POVM
element to one that is independent of the POVM element, by appealing
to a minimax theorem from Game Theory. We show that this minimax
theorem can be applied directly to establish the Quantum Substate
Theorem. The resulting statement is stronger in its dependence on
observational divergence. The original bound read
as~$\sS^\eps_\infty(\rho \| \sigma) \leq d'/\eps - \log_2 (1-\eps)$,
where~$d' \defeq d + 4 \sqrt{d+2} + 2 \log_2 (d+2) + 6$ with~$d \defeq
\sD(\rho\|\sigma)$. The formulation in terms of fidelity also allows
us to show that the dependence on observational divergence in
Theorem~\ref{thm-substate} is optimal up to a constant factor.
\begin{theorem}
\label{thm-converse}
Let~$\cH$ be a Hilbert space, and let $\rho, \sigma \in \rD(\cH)$ be
quantum states such that $\support(\rho) \subseteq \support(\sigma)$.
Suppose~$k \in \reals$ is such that for any~$\eps \in (0,1)$, there is
a quantum state~$\rho'$ with fidelity~$\sF(\rho', \rho) \geq 1 - \eps$
such that~$\rho' \preceq \kappa \sigma$, where
\[
\kappa \quad = \quad \frac{1}{1-\eps}\; 2^{k/\eps}
\enspace, 
\]
or equivalently,
\[
\sS^\eps_\infty(\rho \| \sigma) \quad \leq \quad 
    \frac{k}{\eps} + \log_2 \frac{1}{1-\eps} \enspace.
\]
Then~$\sD(\rho \| \sigma) \leq 4k + 3$.
\end{theorem}
We thus settle two questions posed by Jain
\etal~\cite[Section~5]{JainRS09}.

For the first proof (Section~\ref{sec-proof}), we start by converting
the convex minimization problem (P1) into a min-max problem through a
simple duality argument. The minimax theorem now applies and reduces
the problem of construction of a suitable state~$\rho'$ to one that
works for a fixed POVM element. The latter task turns out to be
similar to proving the Classical Substate Theorem.  This proof is thus
shorter and conceptually simpler than the original one, and also leads
to a tighter dependence on observational divergence. We present a
second proof based on semi-definite programming (SDP) duality
(Section~\ref{sec-sdp}). We believe that both approaches have their
own merits. The first approach is more intuitive in that once the
problem is formulated as a min-max program, the subsequent steps
emerge naturally. The second approach has the appeal of relying on the
more standard SDP duality.  These routes to the theorem may prove
useful in its burgeoning list of applications, as also in the study of
smooth relative min-entropy.

\section{A proof based on min-max duality}
\label{sec-proof}

In this section, we present an alternative proof of the Quantum
Substate Theorem. It hinges on a powerful minimax theorem from game
theory, which is a consequence of the Kakutani fixed point theorem in
real analysis~\cite[Propositions~20.3 and~22.2]{OsborneR94}.
\begin{theorem}
\label{fact:minimax}
Let $A_1,A_2$ be non-empty, convex and compact subsets of $\reals^n$
for some positive integer~$n$. Let $u: A_1 \times A_2 \rightarrow
\reals$ be a continuous function such that
\begin{itemize}
\item $\forall a_2 \in A_2$, the set 
$\{a_1 \in A_1 \;:\; (\forall a_1' \in A_1) \;
               u(a_1,a_2) \geq u(a_1',a_2)\}$ is convex, i.e., for
every~$a_2 \in A_2$, the set of points $a_1 \in A_1$ such that
$u(a_1,a_2)$ is maximum is a convex set; and

\item  $\forall a_1 \in A_1$, the set 
$\{a_2 \in A_2 \;:\; (\forall a_2' \in A_2) \;
                u(a_1,a_2) \leq u(a_1,a_2')\}$ is convex, i.e., for
every $a_1 \in A_1$, the set of points $a_2 \in A_2$, such that
$u(a_1,a_2)$ is minimum is a convex set.
\end{itemize}
Then, there is an $(a_1^\ast, a_2^\ast) \in A_1 \times A_2$ such that 
\begin{displaymath}
\max_{a_1\in A_1}\, \min_{a_2\in A_2} u(a_1,a_2) 
\quad = \quad u(a_1^\ast, a_2^\ast)
\quad = \quad \min_{a_2 \in A_2}\, \max_{a_1 \in A_1} u(a_1,a_2) \enspace.
\end{displaymath}
\end{theorem}

\suppress{The above statement follows by combining Proposition~20.3
  and Proposition~22.2 of Osborne and Rubinstein's book on game
  theory~\cite[pages 19--22]{OsborneR94}.  Proposition~20.3 shows the
  existence of Nash equilibrium (such as~$(a_1^\ast, a_2^\ast)$) in
  strategic games, and Proposition~22.2 connects Nash equilibrium and
  the min-max theorem for games defined using a pay-off function (such
  as $u$).  }

We start with the following lemma which bounds the distance between a
quantum state and its normalized projection onto a subspace in which it has
``large'' support. It is a variant of the ``gentle measurement lemma''
due to Winter~\cite{Winter99}.
\begin{lemma}
\label{lem:eps}
Let $\rho \in \rD(\cH)$ be a quantum state in the Hilbert space
$\cH$. Let $\Pi$ be an orthogonal projection onto a subspace of $\cH$
such that $\tr\, \Pi \rho = \delta < 1$. Let $\rho'' = (\rI - \Pi)
\rho (\rI - \Pi)$ be the projection of~$\rho$ onto the orthogonal
subspace, and let~$\rho' = \frac{\rho''}{\tr \, \rho''}$ be this state
normalized. Then $\sF(\rho,\rho') \geq 1 - \delta$.
\end{lemma}
\begin{proof}
Let $\cK$ be a Hilbert space with $\dim(\cK) =
\dim(\cH)$. Let~$\ket{v} \in \cK \otimes \cH$ be a purification
of~$\rho$~\cite{NielsenC00}.
\suppress{ , that is $\ket{v}$ is a unit vector such that~$\tr_\cK \,
  \density{v} = \rho$ (such a purification always exists
  when~$\dim(\cK) = \dim(\cH)$.)}
Let $\ket{v''} = (\rI \otimes (\rI - \Pi)) \ket{v}$.  Let $\ket{v'} =
\ket{v''}/\norm{ v'' }$.  Observe that $\tr_{\cK}\, \ketbra{v''}{v''}
= \rho''$, so
\[
\norm{v''}^2 \quad = \quad \tr\, \ketbra{v''}{v''} 
    \quad = \quad \tr\, \rho'' \quad = \quad \tr\, \rho - \tr\, \Pi \rho 
    \quad = \quad 1 - \delta
                  \enspace,
\]
and~$\tr_\cK\, \ketbra{v'}{v'} = \rho'$. Now,
\begin{align*}
\sF(\rho, \rho') 
    \quad & \geq \quad \sF(\ketbra{v}{v}, \ketbra{v'}{v'}) \\
    & = \quad |\braket{v}{v'}|^2  \\
    & = \quad  \norm{v''}^2 \quad = \quad 1 - \delta 
          \enspace,
\end{align*}
where the first inequality follows from the monotonicity of fidelity
under completely positive trace preserving (CPTP)
operations~\cite{NielsenC00}.
\end{proof}

The next lemma is an important step in the proof, and along with the
minimax theorem (Theorem~\ref{fact:minimax}) yields the Quantum
Substate Theorem. It mimics the proof of the Classical Substate
Theorem with respect to a particular operator~$M \succeq 0$, which may
be viewed as an unnormalized POVM element. Namely, we decompose~$M$
into its diagonal basis, and imagine measuring with respect to this
basis. If the observational divergence of~$\rho$ with respect
to~$\sigma$ is small, then for most of the basis elements, the
probability of the outcome for~$\rho$ is not too large relative to the
probability for~$\sigma$. Projecting~$\rho$ onto the space spanned by
these basis elements gives us a state~$\rho'$, close to $\rho$, for
which~$\tr\, M \rho'$ is correspondingly bounded, relative to~$\tr\, M
\sigma$.
\begin{lemma}
\label{lem:forM}
Suppose $\rho, \sigma \in \rD(\cH)$ are quantum states in the Hilbert
space~$\cH$ such that $\support(\rho) \subseteq \support(\sigma)$. Let
$d = \sD(\rho \| \sigma)$, $\eps \in (0,1)$, and~$M \succeq 0$ be an
operator on~$\cH$. There exists a quantum state~$\rho' \in \rD(\cH)$
such that $\sF(\rho', \rho) \geq 1 - \eps$ and
\[
(1-\eps) \cdot \tr \, M \rho' 
    \quad \leq \quad 2^{d/\eps} \cdot \tr \, M \sigma \enspace .
\]
\end{lemma}
\begin{proof}
Consider~$M$ in its diagonal form~$M = \sum_{i=1}^{\dim(\cH)} p_i
\ket{v_i}\bra{v_i}$, where the~$(p_i)$ are the eigenvalues of~$M$
corresponding to the orthonormal eigenvectors~$(\ket{v_i})$. Let
\[
B \quad \defeq \quad 
    \set{ i \;:\;  \bra{v_i} \rho \ket{v_i} > 
          2^{d/\eps} \cdot \bra{v_i} \sigma \ket{v_i}, \;
          1 \leq i \leq \dim(\cH) } \enspace .
\]
Let $\Pi = \sum_{i \in B} \density{v_i}$ be the projector onto the
space spanned by vectors specified by $B$. Then $\tr \, \Pi \rho >
2^{d/\eps} \cdot \tr \, \Pi \sigma$ and hence,
\[ 
d \quad \geq \quad \left( \tr\, \Pi \rho \right) 
                    \log_2 \frac{\tr\, \Pi \rho}{\tr\, \Pi \sigma} 
  \quad > \quad \left( \tr \, \Pi \rho \right) \cdot \frac{d}{\eps} 
                 \enspace .
\]
This implies that~$ \tr\, \Pi \rho < \eps$. Let $\rho'' = (\rI - \Pi)
\rho (\rI - \Pi)$ and $\rho' = \frac{\rho''}{\tr {\rho''}} \prec
\frac{\rho''}{1-\eps}$. From Lemma~\ref{lem:eps} we have $\sF(\rho,
\rho') \geq 1 - \eps$. Finally, by the definition of~$\Pi$,
\begin{align*}
(1-\eps) \cdot \tr\, M \rho' 
    \quad & \leq \quad \tr \, M \rho'' \\
          & = \quad \sum_{i \not\in B} p_i \bra{v_i} \rho \ket{v_i} \\
          & \leq \quad 2^{d/\eps} \sum_{i \not\in B} p_i \bra{v_i} 
                       \sigma \ket{v_i} \\
          & \leq \quad 2^{d/\eps} \cdot \tr\, M \sigma \enspace.
\end{align*}
\end{proof}

We now prove the main result, Theorem~\ref{thm-substate}. For this, it
suffices to produce a state close to~$\rho$ that when scaled suitably
is a substate of~$\sigma$.  The condition~$\rho' \preceq \kappa\,
\sigma$ is equivalent to~$\tr\, M\rho' \leq \kappa$ for all~$M \succeq
0$ with $\tr \, M \sigma \leq 1$. We use this dual view of the
substate condition to convert the minimization problem (P1) into a
min-max optimization problem.  We then use the minimax theorem,
Theorem~\ref{fact:minimax}, to drastically simplify the search for a
suitable state~$\rho'$. As a consequence, it suffices to produce a
state~$\rho'$ close to $\rho$ such that~$\tr\, M\rho' \leq \kappa$ for
an arbitrary but fixed~$M \succeq 0$ with $\tr \, M \sigma \leq 1$.

\medskip

\begin{proofof}{Theorem~\ref{thm-substate}}
We first massage the program (P1) into a form to which
Theorem~\ref{fact:minimax} applies.  If a pair~$\rho',\kappa$ are
feasible for (P1), then~$\support(\rho') \subseteq \support(\sigma)$. By
taking~$\cH = \support(\sigma)$ if necessary, we may therefore assume
that~$\sigma \succ 0$, i.e., $\sigma$ has full support.  It is
straightforward to check that for any given $\rho' \in \rD(\cH)$,
\[
\min_{\kappa ~ : ~ \rho' \preceq \, \kappa \sigma } \kappa  
\quad = \quad \max_{M \succeq 0 ~:~ \tr\, M \sigma \, \leq \, 1} \;
\tr\, M \rho' \enspace.
\]
Hence we may rewrite~$\sS^\eps_\infty(\rho \| \sigma)$ as the base-$2$
logarithm of
\[
\min_{\substack{\rho' \succeq 0 ~:~
                \tr \,\rho' =1,\\
                \sF(\rho', \rho) \,\geq\,  1- \eps
               }
} \quad 
\max_{M \succeq 0 ~:~ \tr \, M \sigma \,\leq\, 1 }
\quad \tr\, M \rho' \enspace.
\]
\suppress{
All hypotheses of Theorem~\ref{fact:minimax} are satisfied by the
above optimization problem, except that the feasible region is of the
form~$A_1' \times A_2'$, where~$A_1', A_2'$ are subsets of a
\emph{complex\/} vector space isomorphic to~$\complex^n$, with~$n =
\dim(\cH)^2$. We map this to an optimization problem over the real
vector space~$\reals^m$ with~$m = 2n$ via the natural map~$f :
\complex^n \rightarrow \reals^{m}$ given by $f(a_1+ \complexi b_1,
\dotsc, a_n + \complexi b_n) = (a_1, b_1, \dotsc, a_n, b_n)$. The
map~$f$ is a bijection, is continuous, and has a continuous
inverse. Moreover, $f, f^{-1}$ are both group homomorphisms with
respect to vector addition, and linear with respect to multiplication
by reals.}
Viewing~$\rho'$ and~$M$ as elements of the real vector space of
Hermitian operators in~$\rL(\cH)$, noting that fidelity is concave
in each of its arguments~\cite{NielsenC00} and that the trace function is
bilinear, we may apply
Theorem~\ref{fact:minimax} to the resulting optimization problem. We
get
\begin{align*}
\lefteqn{2^{\sS^\eps_\infty(\rho \| \sigma)}} \\
& \quad = \quad 
\max_{M \,\succeq\, 0 ~:~ \tr \, M \sigma \,\leq\, 1} 
\quad 
\min_{\substack{ \rho' \,\succeq\, 0 ~:~
                 \tr \,\rho' \,=\, 1, \\ 
                 \sF(\rho', \rho) \,\geq\, 1 - \eps
               }
} \quad \tr \, M \rho' \enspace .
\end{align*}
By Lemma~\ref{lem:forM}, for every~$M \succeq 0$ with~$\tr\, M\sigma
\leq 1$, there is a quantum state~$\rho'$, with $\sF(\rho',\rho) \geq
1 - \eps$, such that~$(1 - \eps) \, \tr\, M \rho' \leq 2^{d/\eps}$,
where~$d = \sD(\rho \| \sigma)$.  The desired result now follows.
\end{proofof}

Combining Theorem~\ref{thm-substate} and the Uhlmann
theorem~\cite{NielsenC00} immediately gives us the following
statement. The Quantum Substate Theorem is often used in this form in
its applications.
\begin{corollary} 
Let~$\cH, \cK$ be Hilbert spaces with~$\dim(\cK) \geq \dim(\cH)$, and
let $\rho, \sigma \in \rD(\cH)$ be quantum states such that
$\support(\rho) \subseteq \support(\sigma)$.  Let $d = \sD(\rho \|
\sigma)$, $\eps \in (0,1)$, and~$\ket{v} \in \cK \otimes \cH$ be a
purification of~$\rho$. Then there is a pure state~$\ket{v'} \in \cK
\otimes \cH$ with $\sF(\density{v}, \density{v'}) \geq 1 - \eps$, and
a pure state~$\ket{w'} \in \cK \otimes \cH$ such that
$\ket{w} \in \mathbb{C}^2 \otimes \cK \otimes \cH$ defined as
\[
\ket{w} \quad \eqdef \quad 
    \sqrt{\alpha} \; \ket{0}\ket{v'} + \sqrt{1-\alpha}\; \ket{1}\ket{w'}
    \enspace,
\]
with~$\alpha = (1-\eps)2^{-d/\eps}$, is a purification of~$\sigma$.
\end{corollary}
\begin{proof}
Let~$\rho'$ be a state given by Theorem~\ref{thm-substate} such that
fidelity~$\sF(\rho', \rho) \geq 1 - \eps$ and~$\alpha \rho' \preceq
\sigma$.  Then we can decompose~$\sigma$ as
\[
\sigma \quad = \quad  \alpha \rho'  + (1 - \alpha) \theta \enspace ,
\]
where~$\theta \in \sD(\cH)$ is some quantum state. By the Uhlmann
theorem~\cite{NielsenC00} there is a purification~$\ket{v'} \in \cK
\otimes \cH$ of $\rho'$ such that $\sF(\density{v}, \density{v'}) =
\sF(\rho, \rho') \geq 1 - \eps$. Let $\ket{w'} \in \cK \otimes \cH$ be
any purification of $\theta$. Then we may verify that~$\ket{w}$ as
defined in the statement of the corollary is a purification of
$\sigma$.
\end{proof}

The dependence of the bound on the~$\eps$-smooth relative min-entropy
in Theorem~\ref{thm-substate} in terms of observational divergence is
optimal up to a constant factor, as stated in
Theorem~\ref{thm-converse}. We start its proof with the following
lemma.

\begin{lemma}
\label{lem-prob}
Let~$\delta, \delta' \in [0,1]$ and~$\beta \in [0,1/4]$ such that
\[
\left( 
    \sqrt{\delta\, \delta'} + \sqrt{(1 - \delta)(1 - \delta')}
\right)^2
\quad \geq \quad 1 - \beta \delta \enspace.
\]
Then~$\delta' \geq \left( 1-\sqrt{\beta}\, \right)^2 \delta$.
\end{lemma}
\begin{proof}
Let~$u = \left( \sqrt{\delta}, \sqrt{1-\delta} \right)^\transpose$
and~$u' = \left( \sqrt{\delta'}, \sqrt{1-\delta'} \right)^\transpose$
be vectors in~$\reals^2$.  Let~$\phi, \phi' \in [0,\pi/2]$ be the
angles~$u,u'$ make with~$(1,0)^\transpose$, respectively. By
hypothesis, $\braket{u}{u'}^2 \ge 1 - \beta \delta$. Let~$\theta \in
[0,\pi/2]$ be the angle between~$u,u'$, so that~$\cos^2 \theta \ge 1 -
\beta \delta$.

We wish to bound~$\delta' = \cos^2 \phi'$ from below given
that~$\size{\phi' - \phi} = \theta$. Observe that~$\cos(\phi+\theta)
\ge \sqrt{\delta (1-\beta\delta)} - \sqrt{(1 - \delta)\beta\delta} \ge
0$, so that~$\phi + \theta \le \pi/2$. Therefore, $\delta'$ takes its
minimum value when~$\phi' = \phi+\theta$.

We may now bound~$\delta'$ as follows.
\begin{eqnarray*}
\delta' \quad = \quad \cos^2 \phi' 
    & \ge &  \cos^2 (\phi+\theta) \\
    & \ge & \left( \sqrt{\delta (1-\beta\delta)} 
             - \sqrt{(1 - \delta)\beta\delta} \right)^2 \\
    & =   & (1+\beta) \delta - 2 \beta \delta^2 - 2 \sqrt{\beta}\, \delta 
            \left( 1 - (1+\beta) \delta + \beta \delta^2 \right)^{1/2} \\
    & \ge & (1+\beta) \delta - 2 \beta \delta^2 - 2 \sqrt{\beta}\, \delta 
            (1 - \delta)^{1/2} \\
    & \ge & (1+\beta) \delta - 2 \beta \delta^2 - 2 \sqrt{\beta}\, \delta 
            (1 - \delta/2) \\
    & =   & \left(1- \sqrt{\beta}\right)^2 \delta 
            + (\sqrt{\beta} - 2 \beta) \delta^2 \\
    & \ge & \left(1- \sqrt{\beta}\right)^2 \delta \enspace,
\end{eqnarray*}
since~$\beta \leq 1/4$.
\end{proof}

We are now ready to prove the optimality of Theorem~\ref{thm-substate}.

\medskip

\begin{proofof}{Theorem~\ref{thm-converse}}
It suffices to prove that for any POVM element~$M \in \rP(\cH)$
with~$\tr(M\rho) \neq 0$,
\[
\tr(M\rho) \; \log \frac{\tr(M\rho)}{\tr(M\sigma)}
\]
is bounded by~$4k + 3$ from above. 

Fix such a POVM element~$M$, let~$\delta = \tr(M\rho)$, and~$\eps =
\beta \delta$ for some~$\beta \in (0,1)$ to be specified later. By
hypothesis, there is a quantum state~$\rho' \in \rD(\cH)$
with~$\sF(\rho',\rho) \ge 1 - \eps$ and~$\rho' \preceq \kappa\, \sigma$,
where
\[
\kappa \quad = \quad \frac{2^{k/\eps}}{1-\eps} \enspace.
\]
Let~$\delta' = \tr(M\rho')$. By the monotonicity of fidelity under
CPTP operations~\cite{NielsenC00}, we have
\[
\left(
    \sqrt{\delta\, \delta'} + \sqrt{(1 - \delta)(1 - \delta')}
\right)^2
\quad \ge \quad \sF(\rho',\rho) \quad \ge \quad 1 - \eps 
\quad = \quad 1 - \beta \delta \enspace.
\]
By Lemma~\ref{lem-prob}, we have~$\tr(M\rho') = \delta' \ge
(1-\sqrt{\beta}\,)^2 \delta$ if~$\beta \leq 1/4$.

We set~$\beta = 1/4$, so that~$\tr(M\rho') \geq \delta/4$.
Furthermore,
\[
\tr(M\sigma) \quad \geq \quad \frac{(1-\eps)}{2^{k/\eps}} \; \tr(M\rho')
\quad \geq \quad \frac{(1-\delta/4)}{2^{4k/\delta}} (\delta/4)
\quad \geq \quad \frac{\delta}{ 2^{3 + 4k/\delta}} \enspace,
\]
as~$\delta \leq 1$.
So
\[
\tr(M\rho) \; \log \frac{\tr(M\rho)}{\tr(M\sigma)}
\quad = \quad \delta \; \log \frac{\delta}{\tr(M\sigma)}
\quad \leq \quad 4k + 3 \enspace.
\]
\end{proofof}

\section{A proof based on SDP duality}
\label{sec-sdp}

In this section we present a second alternative proof of the Quantum
Substate Theorem, Theorem~\ref{thm-substate}. The proof is based on a
formulation of smooth relative min-entropy as a semi-definite
program. 

The optimization problem~(P1) in Section~\ref{sec-substate} is seen to
be an SDP once we express the fidelity constraint as a semi-definite
inequality.  This is based on a formulation due to
Watrous~\cite{Watrous11} of the fidelity of two quantum states as an
SDP. For completeness, we include a proof of its correctness.

\begin{lemma}[Watrous]
\label{lem-fidelity}
Suppose $\rho, \rho' \in \rD(\cH)$ are quantum states in the Hilbert
space~$\cH$. The fidelity~$\sF(\rho,\rho')$ of the two states equals
the square of the optimum of the following SDP over the variable~$X
\in \rL(\cH)$.
\begin{align*}
\text{maximize:} & \quad \frac{1}{2} \left( \tr\, X 
                         + \tr\, X^\adjoint \right) \\
   \text{subject to:} & \\
   & \left( \begin{array}{cc}
               \rho' & X \\
               X^\adjoint & \rho
            \end{array}
     \right) \quad \succeq \quad 0 & \mathrm{(P2)} \\
   & X \in \rL(\cH)
\end{align*}
\end{lemma}
\begin{proof}
By Theorem~IX.5.9 in the text~\cite{Bhatia97}, the matrix inequality
in the program~(P2) holds iff there is an operator~$Y \in \rL(\cH)$ such
that~$\norm{Y} \leq 1$ and~$X = \sqrt{\rho'}\, Y \sqrt{\rho}$.
Since~$\sF(\rho',\rho) = \trnorm{ \sqrt{\rho'} \sqrt{\rho} }^2$ and we
may characterize trace norm as~$\trnorm{M} = \max \set{ \size{\tr(ZM)}
  \;:\; Z \in \rL(\cH), \norm{Z} \le 1 }$ for any~$M \in \rL(\cH)$,
the lemma follows.
\end{proof}

The problem~(P1) may now be formulated as the following SDP with
variables~$\kappa \in \reals, \rho' \in \rL(\cH), X \in \rL(\cH)$ in
the primal problem, and variables~$Z_1, Z_2 \in \rL(\cH)$ and~$z_3,
z_4 \in \reals$ in the dual, where~$Z_1, Z_2$ are Hermitian.
\begin{center}
  \begin{minipage}[t]{3in}
    \centerline{\underline{P3 Primal problem}}\vspace{-5mm}
\begin{align*}
  \text{minimize:} & \quad \kappa \\
  \text{subject to:}   & \\
    \rho' & \quad \preceq \quad \kappa\, \sigma \\
    \tr \, \rho' & \quad = \quad 1 \\
    \left( \begin{array}{cc}
               \rho' & X \\
               X^\adjoint & \rho
            \end{array}
    \right) & \quad \succeq \quad 0 \\
    \tr\, X + \tr\, X^\adjoint & \quad \geq \quad 2 \sqrt{1-\eps} \\
    \rho' \in \rL(\cH), & \quad \rho' \succeq 0 \\
    \kappa \in \reals, & \quad \kappa \geq 0 \\
    X \in \rL(\cH)
\end{align*}
\end{minipage}
\vspace{1em}
\begin{minipage}[t]{3in}
    \centerline{\underline{P3 Dual problem}}\vspace{-5mm}
    \begin{align*}
      \text{maximize:} \quad z_4 + 2 z_3 \sqrt{1-\eps} & + \tr(Z_2 \rho) \\ 
      \text{subject to:} \hspace{7em} & \\
          \tr(Z_1 \sigma) & \quad \leq \quad 1 \\
          \left( \begin{array}{cc}
               z_4 \rI - Z_1 & z_3 \rI \\
               z_3 \rI & Z_2
            \end{array}
    \right) & \quad \preceq \quad 0 \\
    Z_1 \in \rL(\cH), & \qquad Z_1 \succeq 0 \\
    z_3, z_4 \in \reals, & \qquad  z_3 \geq 0 \\
    Z_2 \in \rL(\cH), & \qquad Z_2 \textrm{ Hermitian}
    \end{align*}
\end{minipage}
\end{center}
The equivalence of the problems~(P1) and~(P3) follows from
Lemma~\ref{lem-fidelity} and paves the way for the second proof.

\medskip
\begin{proofof}{Theorem~\ref{thm-substate}}
We may verify that strong duality holds since the P3 primal program is
feasible, and the dual is strictly
feasible~\cite{Watrous11,Lovasz03}. Therefore, it suffices to bound the dual
objective function for any set of dual feasible variables~$(Z_1, Z_2,
z_3, z_4)$.

By Lemma~\ref{lem:forM}, there is a quantum state~$\rho'$,
with~$\sF(\rho',\rho) \geq 1 - \eps$, such that~$(1 - \eps) \, \tr
(Z_1 \rho') \leq 2^{d/\eps}\, \tr(Z_1 \sigma) \leq 2^{d/\eps}$, where~$d
= \sD(\rho \| \sigma)$.

Since~$\sF(\rho',\rho) \geq 1 - \eps$, by Lemma~\ref{lem-fidelity},
there is an operator~$X \in \rL(\cH)$ such that
\[
\left( \begin{array}{cc}
               \rho' & X \\
               X^\adjoint & \rho
            \end{array}
    \right)  \quad \succeq \quad 0 \enspace,
\]
and~$\tr\, X + \tr\, X^\adjoint \geq 2 \sqrt{1-\eps}$. Therefore,
\[
\tr \left( \begin{array}{cc}
               \rho' & X \\
               X^\adjoint & \rho
            \end{array}
\right)
\left( \begin{array}{cc}
               z_4 \rI - Z_1 & z_3 \rI \\
               z_3 \rI & Z_2
            \end{array}
\right)
    \quad \le \quad 0 \enspace,
\]
In other words,
\begin{eqnarray*}
z_4 - \tr(Z_1 \rho') + z_3 (\tr\, X + \tr\, X^\adjoint) +  \tr(Z_2 \rho) 
    & \le &  0 \enspace,
\end{eqnarray*}
which implies that the dual objective function is bounded as
\begin{eqnarray*}
z_4 + 2 z_3 \sqrt{1-\eps} +  \tr(Z_2 \rho) 
    & \le &  \tr(Z_1 \rho') \quad \le \frac{2^{d/\eps}}{1-\eps}
             \enspace.
\end{eqnarray*}
This completes the proof.
\end{proofof}

\section{Conclusion}

We presented two alternative proofs of the Quantum Substate Theorem
due to Jain, Radhakrishnan, and Sen~\cite{JainRS02,JainRS09}. In
addition to giving bounds on the smooth relative min-entropy of two
quantum states, this gives us a powerful operational interpretation of
relative entropy and observational divergence.  In the process, we
resolve two questions left open by Jain \etal.

The crucial insight here is that the we may express smooth relative
min-entropy as a convex or semi-definite program and appeal to duality
theory. In this respect, we join a growing number of applications of
convex and semi-definite programming to quantum information processing.
This approach can be extended to the more general notion of
smooth relative min-entropy studied by Renner~\cite{Renner05} to get
similar bounds on this quantity. This view of the quantity may shed
light on its numerous applications.

\subsection*{Acknowledgements}

For helpful discussions and pointers to the literature, we thank
Matthias Christandl, Rajat Mittal, Marco Piani, Renato Renner, Pranab
Sen, and John Watrous.

Our first attempt at this work involved an awkward formulation of
smooth relative min-entropy defined with trace distance as a
semi-definite program. We set aside this formulation in favour of the
min-max program, which we believe also leads to a more intuitive
proof.  We thank John Watrous for sharing a simple SDP formulation (with
a bound on trace distance) and the resulting proof. Our current presentation 
involving fidelity supersedes this.

Finally, we thank the referees for their extensive feedback. It lead 
not only to
what we hope is a vastly improved presentation of the work, but also
to the discovery of a tight connection between smooth relative
min-entropy and observational divergence.

\bibliography{substate-new-proof}
\bibliographystyle{plain}

\end{document}